%% file: spanners_with_local_error_arxv.tex
\documentclass[runningheads]{llncs}
\usepackage{amsmath, amsfonts}
\usepackage{xcolor}
\usepackage{comment}
\usepackage{graphicx,caption}
\usepackage[ruled,vlined]{algorithm2e}
\usepackage{tikz}
\usetikzlibrary{decorations.pathreplacing}
\usepackage[noadjust]{cite}

\DeclareMathOperator{\dist}{dist}

\newcommand{\Pc}{\mathcal{P}}
\newcommand{\Sc}{\mathcal{S}}
\newcommand{\eps}{\varepsilon}
\newcommand{\Oish}{\widetilde{O}}
\DeclareMathOperator{\MST}{MST}

\usepackage[colorinlistoftodos,textsize=small,color=blue!25!white,obeyFinal]{todonotes}
\usepackage{thmtools, thm-restate}

\begin{document}
\title{On additive spanners in weighted graphs with local error}
\author{Reyan~Ahmed\inst{1} \and
Greg~Bodwin\inst{2} \and 
Keaton~Hamm\inst{3} \and
Stephen~Kobourov\inst{1} \and 
Richard~Spence\inst{1}}
\institute{Department of Computer Science, University of Arizona \and
Department of Computer Science, University of Michigan \and
Department of  Mathematics, University of Texas at Arlington}
\date{}
\authorrunning{Ahmed et al.}
\maketitle 
\begin{abstract}
    An \emph{additive $+\beta$ spanner} of a graph $G$ is a subgraph which preserves distances up to an additive $+\beta$ error. Additive spanners are well-studied in unweighted graphs but have only recently received attention in weighted graphs [Elkin et al.\ 2019 and 2020, Ahmed et al.\ 2020]. This paper makes two new contributions to the theory of weighted additive spanners.
    
    For weighted graphs, [Ahmed et al.\ 2020] provided constructions of sparse spanners with \emph{global} error $\beta = cW$, where $W$ is the maximum edge weight in $G$ and $c$ is constant.
    We improve these to \emph{local} error by giving spanners with additive error $+cW(s,t)$ for each vertex pair $(s,t)$, where $W(s, t)$ is the maximum edge weight along the shortest $s$--$t$ path in $G$. These include pairwise $+(2+\eps)W(\cdot,\cdot)$ and $+(6+\eps) W(\cdot, \cdot)$ spanners over vertex pairs $\Pc \subseteq V \times V$ on $O_{\eps}(n|\Pc|^{1/3})$ and $O_{\eps}(n|\Pc|^{1/4})$ edges for all $\eps > 0$, which extend previously known unweighted results up to $\eps$ dependence, as well as an all-pairs $+4W(\cdot,\cdot)$ spanner on $\Oish(n^{7/5})$ edges.
    
    Besides sparsity, another natural way to measure the quality of a spanner in weighted graphs is by its \emph{lightness}, defined as the total edge weight of the spanner divided by the weight of an MST of $G$.
    We provide a $+\eps W(\cdot,\cdot)$ spanner with $O_{\eps}(n)$ lightness, and a $+(4+\eps) W(\cdot,\cdot)$ spanner with $O_{\eps}(n^{2/3})$ lightness. These are the first known additive spanners with nontrivial lightness guarantees. All of the above spanners can be constructed in polynomial time.
\end{abstract}

\section{Introduction}

Given an undirected graph $G(V,E)$, a \emph{spanner} is a subgraph $H$ which approximately preserves distances in $G$ up to some error.
Spanners are an important primitive in the literature on network design and shortest path algorithms, with applications in motion planning in robotics~\cite{MB13,DB14,DBLP:journals/ijcga/CaiK97,SSAH14}, asynchronous protocol design~\cite{PU89jacm}, approximate shortest path algorithms~\cite{dor2000all}, and much more; see survey \cite{ahmed2020graphElsevier}.
One general goal in research on spanners is to minimize the size or sparsity of the spanner (as measured by the number of edges $|E(H)|$), given some error by which distances can be distorted.
For weighted graphs, another desirable goal is to minimize the \emph{lightness}, defined as the total weight of the spanner divided by the weight of a minimum spanning tree (MST) of $G$.

Spanners were introduced in the 1980s by Peleg and Sch\"{a}ffer~\cite{doi:10.1002/jgt.3190130114}, who first considered \emph{multiplicative} error. A subgraph $H$ is a (multiplicative) $k$-spanner of $G(V,E)$ if $d_H(s,t) \le k \cdot d_G(s,t)$ for all vertices $s,t \in V$, where $d_G(s,t)$ is the distance between $s$ and $t$ in $G$. Since $H$ is a subgraph, we also have $d_G(s, t) \le d_H(s, t)$ by definition, so the distances in $H$ approximate those of $G$ within a multiplicative factor of $k$, sometimes called the stretch factor. Alth{\" o}fer et al.~\cite{Alth90} showed that all $n$-vertex graphs have multiplicative $(2k-1)$-spanners on $O(n^{1+1/k})$ edges with $O(n/k)$ lightness, and this edge bound is the best possible assuming the girth conjecture by Erd\H{o}s~\cite{erdos1963extremal} from extremal combinatorics. Meanwhile, this initial lightness bound has been repeatedly improved in follow-up work, and the optimal bound still remains open \cite{CDNS92,ENS14,CW16,Filtser16,LS20}.

Multiplicative spanners are extremely well applied in computer science.
However, they are typically applied to very large graphs where it may be undesirable to take on errors that scale with the (possibly very large) distances in the input graph.
A more desirable error is \emph{additive} error which does not depend on the original graph distances at all:
\begin{definition}[Additive $+\beta$ spanner]
Given a graph $G(V,E)$ and $\beta \ge 0$, a subgraph $H$ is a $+\beta$ spanner of $G$ if \begin{equation} \label{eqn:spanner-ineq}
    d_G(s, t) \le d_H(s, t) \le d_G(s, t) + \beta
\end{equation}
for all vertices $s,t \in V$.
\end{definition}

A pairwise $+\beta$ spanner is a subgraph $H$ for which~\eqref{eqn:spanner-ineq} only needs to hold for specific vertex pairs $\Pc \subseteq V \times V$ given on input, and a subsetwise spanner is a pairwise spanner with $\Pc = \Sc \times \Sc$ for some $\Sc \subseteq V$ (if $\Pc = V \times V$, these are sometimes called all-pairs spanners, for clarity).
It is known that all unweighted graphs $G(V,E)$ with $|V|=n$ have (all-pairs) $+2$ spanners on $O(n^{3/2})$ edges \cite{Aingworth99fast,knudsen2014additive},
$+4$ spanners on $\Oish(n^{7/5})$ edges \cite{chechik2013new,bodwin2020note}, and $+6$ spanners on $O(n^{4/3})$ edges \cite{baswana2010additive,knudsen2014additive,woodruff2010additive}. On the negative side, there exist graphs which have no $+\beta$ spanner on $O(n^{4/3 - \eps})$ edges even for arbitrarily large constant $\beta$~\cite{abboud20174}. This presents a barrier to using additive spanners in applications where a very sparse subgraph, say on $O(n^{1.001})$ edges, is needed. However, the lower bound construction in~\cite{abboud20174} is rather pathological; tradeoffs do continue for certain natural classes of graphs with good girth or expansion properties~\cite{baswana2010additive}, and recent experimental work~\cite{ahmed2021multi} showed that tradeoffs seem to continue for graphs constructed from common random graph models.

A more serious barrier preventing the applicability of additive spanners is that classic constructions only apply to \emph{unweighted} graphs, while many naturally-occurring metrics are not expressible by a unit-weight graph.
To obtain additive spanners of weighted graphs $G=(V,E,w)$ where $w:E \to \mathbb{R}^+$, the error term $+\beta$ needs to scale somehow with the edge weights of the input graph.
Prior work \cite{ahmed2020weighted} has considered \emph{global} error of type $\beta = cW$, where $W = \max_{e \in E} w(e)$ is the maximum edge weight in the input graph and $c$ is a constant.
However, a more desirable paradigm studied by Elkin, Gitlitz, and Neiman~\cite{elkin2019almost,elkin2020improved} is to consider \emph{local} error in terms of the maximum edge weight along a shortest $s$--$t$ path:
\begin{definition} [Local $+cW(\cdot,\cdot)$ spanner]
Given a graph $G(V,E)$, subgraph $H$ is a (local) $+cW(\cdot,\cdot)$ spanner if $d_H(s, t) \le d_G(s, t) + cW(s, t)$ for all $s, t \in V$, where $W(s, t)$ is the maximum edge weight along a shortest path $\pi(s,t)$ in $G$.\footnote{If there are multiple shortest $s$--$t$ paths, then 
we break ties consistently so that subpaths of shortest paths are also shortest paths.}
\end{definition}

It is often the case that $W(s,t) \ll W$ for many vertex pairs $(s,t)$ in which a $+cW(\cdot,\cdot)$ spanner has much less additive error for such vertex pairs.
Additionally, a $+cW(\cdot,\cdot)$ spanner is also a multiplicative $(c+1)$-spanner, whereas a $+cW$ spanner can have unbounded multiplicative stretch.
This relationship between additive and multiplicative stretch is thematic in the area \cite{Elkin:2004:SCG:976327.984900,baswana2010additive}.


\paragraph{Sparse Local Additive Spanners.}
Weighted additive spanners were first studied by Elkin et al.~\cite{elkin2019almost}, who gave a local $+2W(\cdot,\cdot)$ spanner with $O(n^{3/2})$ edges, as well as a ``mixed'' spanner with a similar error type. Ahmed, Bodwin, Sahneh, Kobourov, and Spence~\cite{ahmed2020weighted} gave a comprehensive study of weighted additive spanners, including a \emph{global} $+4W$ spanner with $\Oish(n^{7/5})$ edges and $+8W$ spanner with $O(n^{4/3})$ edges, analogous to the previously-known unweighted constructions.
The $+6$ unweighted error vs.\ $+8W$ global weighted error left a gap to be closed; this was mostly closed in a recent follow-up work of Elkin et al.~\cite{elkin2020improved}, who gave a \emph{local} all-pairs $+(6+\eps)W(\cdot,\cdot)$ spanner on $O_{\eps}(n^{4/3})$ edges\footnote{We use $O_{\eps}(f(n))$ as shorthand for $O(\text{poly}(\frac{1}{\eps})f(n))$.} by generalizing the $+6$ spanner by Knudsen~\cite{knudsen2014additive}, and a subsetwise $+(2+\eps)W(\cdot,\cdot)$ spanner on $O_{\eps}(n\sqrt{|\Sc|})$ edges~\cite{elkin2020improved} (following a similar $+2$ subsetwise spanner in unweighted graphs by Elkin (unpublished), later published in~\cite{Pettie09,Cygan13}).
Our first contribution is the improvement of several remaining known constructions of weighted additive spanners from global to local error.


\begin{restatable}{thm}{sparse}
\label{thm:sparse}
Let $\eps > 0$. Then every weighted graph $G$ and set $\Pc \subseteq V \times V$ of vertex pairs has:
\begin{enumerate}\setlength{\itemsep}{2pt}
    \item a deterministic pairwise $+(2+\eps)W(\cdot,\cdot)$ spanner on $O_{\eps}(n|\Pc|^{1/3})$ edges, \label{thm:pairwise-2epsW}
    \item a deterministic pairwise $+(6+\eps)W(\cdot,\cdot)$ spanner on $O_{\eps}(n|\Pc|^{1/4})$ edges, \label{thm:pairwise-6epsW}
    \item a pairwise $+2W(\cdot,\cdot)$ spanner on $O(n|\Pc|^{1/3})$ edges, \label{thm:pairwise-2W}
    \item a pairwise $+4W(\cdot,\cdot)$ spanner on $O(n|\Pc|^{2/7})$ edges, and \label{thm:pairwise-4W}
    \item an all-pairs $+4W(\cdot,\cdot)$ spanner on $\Oish(n^{7/5})$ edges. \label{thm:allpairs-4W}
\end{enumerate}
\end{restatable}

Unweighted versions of these results were proved in \cite{Kavitha15,Kavitha2017,bodwin2020note}, and weighted versions with global error but without $\eps$ dependence were proved in \cite{ahmed2020weighted}.
This theorem is the first to provide versions with local error. Moreover, the first two pairwise constructions are deterministic, unlike the randomized constructions from~\cite{ahmed2020graphElsevier}.
Together with \cite{elkin2019almost,elkin2020improved}, the above results complete the task of converting unweighted additive spanners to weighted additive spanners with local error; see Tables~\ref{table:additive} and~\ref{table:additive-pairwise}.

\begin{table}[h]
\centering
\begin{tabular}{|c|c|c||c|c|c|} \hline
\multicolumn{3}{|c}{Unweighted} & \multicolumn{3}{c|}{Weighted} \\ \hline
$+\beta$ & Size & Ref. & $+\beta$ & Size & Ref. \\ \hline
$+2$ & $O(n^{3/2})$ & \cite{Aingworth99fast, baswana2010additive, Cygan13, knudsen2014additive} & $+2W(\cdot, \cdot)$ & $O(n^{3/2})$ & \cite{elkin2019almost} \\ \hline
$+4$ & $\Oish(n^{7/5})$ & \cite{chechik2013new} & $+4W(\cdot,\cdot)$ & $\Oish(n^{7/5})$ & \textbf{[this paper]} \\ \hline
$+6$ & $O(n^{4/3})$ & \cite{baswana2010additive, knudsen2014additive} & $+(6+\eps)W(\cdot, \cdot)$ & $O_{\eps}(n^{4/3})$ & \cite{elkin2020improved} \\ \hline
$+n^{o(1)}$ & $\Omega(n^{4/3 - \eps})$ & \cite{abboud20174} &  &  & \\ \hline
\end{tabular}
\caption{All-pairs additive spanner constructions for unweighted and weighted graphs.}
\label{table:additive}
\end{table}

\begin{table}[h]
\centering
\begin{tabular}{|l|c|c|c||c|c|c|} \hline
\multicolumn{4}{|c}{Unweighted} & \multicolumn{3}{c|}{Weighted} \\ \hline
Type & $+\beta$ & Size & Ref. & $+\beta$ & Size & Ref. \\ \hline
Subset & $+2$ & $O(n\sqrt{|\Sc|})$ & \cite{Cygan13} & $+(2+\eps)W(\cdot,\cdot)$ & $O_{\eps}(n\sqrt{|\Sc|})$ & \cite{elkin2020improved} \\ \hline
Pairwise & $+2$ & $O(n|\Pc|^{1/3})$ & \cite{Kavitha15,censor2016distributed} & $+2W(\cdot,\cdot)$ & $O(n|\Pc|^{1/3})$ & \textbf{[this paper]} \\ \hline
Pairwise & $+4$ & $\Oish(n|\Pc|^{2/7})$ & \cite{Kavitha2017} & $+4W(\cdot,\cdot)$ & $O(n|\Pc|^{2/7})$ & \cite{ahmed2021weighted} \\ \hline
Pairwise & $+6$ & $O(n|\Pc|^{1/4})$ & \cite{Kavitha2017} & $+(6+\eps)W(\cdot,\cdot)$ & $O_{\eps}(n|\Pc|^{1/4})$ & \textbf{[this paper]} \\ \hline
\end{tabular}
\caption{Pairwise and subsetwise (purely) additive spanner constructions for unweighted and weighted graphs.}
\label{table:additive-pairwise}
\end{table}

\paragraph{Lightweight Local Additive Spanners.} All of the aforementioned results are in terms of the number of edges $|E(H)|$ of the spanner. If minimizing the total edge weight is more desirable than constructing a sparse spanner, a natural problem is to construct \emph{lightweight} spanners. Given a connected graph $G=(V,E)$ with positive edge weights $w:E \to \mathbb{R}^+$, the \emph{lightness} of a subgraph $H$ is defined by
\begin{equation}\label{eqn:lightness}
    \text{lightness}(H) := \frac{w(H)}{w(\MST(G))}
\end{equation}
where $w(H)$ and $w(\MST(G))$ are the sum of edge weights in $H$ and an MST of $G$, respectively. Section~\ref{sec:lightweight} highlights why none of the aforementioned sparse spanners have good lightness guarantees. Our second contribution is the following:
\begin{restatable}{thm}{lightness} \label{thm:lightness}
Let $\eps > 0$. Then every weighted graph $G$ has:
\begin{enumerate}\setlength{\itemsep}{2pt}
    \item a deterministic all-pairs $+\eps W(\cdot,\cdot)$ spanner with $O_{\eps}(n)$ lightness, and \label{thm:lightness-epsW}
    \item a deterministic all-pairs $+(4+\eps)W(\cdot,\cdot)$ spanner with $O_{\eps}(n^{2/3})$ lightness. \label{thm:lightness-4epsW}
\end{enumerate}
\end{restatable}
To the best of our knowledge, these are the first nontrivial lightness results known for additive spanners.
For comparison on the first result, it is easy to show that every graph $G$ has an all-pairs distance preserver $H$ with $\text{lightness}(H) = O(n^2)$.
It follows from the seminal work of Khuller, Raghavachari, and Young \cite{khuller1995balancing} on shallow-light trees that every graph $G$ has a subgraph $H$ that preserves distances up to a $(1+\eps)$ multiplicative factor, with lightness $O_{\eps}(n)$.
Our first result implies that the same lightness bound (up to the specifics of the $\eps$ dependence) can be achieved with \emph{additive} error. Theorem~\ref{thm:lightness}.\ref{thm:lightness-epsW} strictly strengthens this consequence of \cite{khuller1995balancing}, since local $+\eps W(\cdot,\cdot)$ error implies multiplicative $(1+\eps)$ stretch as previously mentioned.
Additionally, Theorem~\ref{thm:lightness}.\ref{thm:lightness-epsW} is tight, in the sense that an unweighted complete graph $K_n$ has lightness $\Theta(n)$ and no nontrivial $+\eps W(\cdot, \cdot)$ spanner. Table~\ref{table:additive-lightweight} summarizes these lightness results compared with those in unweighted graphs.

\begin{table}
\centering
\begin{tabular}{|c|c||c|c|c|} \hline
\multicolumn{2}{|c||}{Unweighted} &
\multicolumn{3}{c|}{Weighted} \\ \hline
$+\beta$ & Lightness & $+\beta$ & Lightness & Ref. \\ \hline
0 & $O(n)$ & $+\eps W(\cdot,\cdot)$ & $O_{\eps}(n)$ & \textbf{[this paper]} \\ \hline
$+2$ & $O(n^{1/2})$ &  & ? & \\ \hline
$+4$ & $\Oish(n^{2/5})$ & $+(4+\eps)W(\cdot,\cdot)$ & $O_{\eps}(n^{2/3})$ & \textbf{[this paper]} \\ \hline
$+6$ & $O(n^{1/3})$ & & ? & \\ \hline
\end{tabular}
\caption{Lightweight all-pairs additive spanners in unweighted and weighted graphs.}
\label{table:additive-lightweight}
\end{table}

The lightweight spanners are based on a new initialization technique which we call $d$-\emph{lightweight initialization}, in which an initial set of lightweight edges is added to the spanner starting from the MST of the input graph. Nearly all of the above spanners have a common theme in the construction and analysis: add an initial set of edges oblivious to the distances in the graph, then add shortest paths for any vertex pairs which do not satisfy inequality~\eqref{eqn:spanner-ineq}. The size or lightness bounds are then analyzed by determining how many pairs of nearby vertices there are whose distances sufficiently improve upon adding a shortest path; this method was also used by Elkin et al.~\cite{elkin2020improved}. Note that such improvements in unweighted graphs must be by at least 1; in weighted graphs, distances may improve by arbitrarily small amount which leads to the $\eps$ dependence. 
We leave as open questions the lightness bounds for $+2W$ and $+6W$ spanners (with or without $\eps$ dependence), whether the lightness for the $+(4+\eps)W(\cdot,\cdot)$ spanner can be improved to $\Oish(n^{2/5})$, and whether we can construct additive spanners in weighted graphs which are simultaneously sparse \emph{and} lightweight.

\section{Preliminaries}
Many additive spanner constructions begin with either a clustering~\cite{baswana2010additive,Cygan13,Kavitha2017} or initialization phase~\cite{knudsen2014additive,ahmed2020weighted,elkin2020improved}, where an initial set of edges is added to the spanner oblivious to distances or vertex pairs $\Pc$ in the graph; after this phase, additional edges or paths are added so that the resulting subgraph is a valid spanner. Experimental results~\cite{ahmed2021multi} suggest that initialization is preferred over clustering in terms of runtime and spanner size, and all constructions in this paper are initialization-based.
For weighted graphs, a $d$-\emph{light initialization} of $G$ is a subgraph obtained by selecting the $d$ lightest edges incident to every vertex, or all edges if the degree is less than $d$. We exploit the following lemma from~\cite{ahmed2020weighted}:

\begin{lemma}[\cite{ahmed2020weighted}]
\label{lemma:d_initialize}
Let $H$ be a $d$-light initialization
of a weighted graph $G$, and let $\pi(s,t)$ be a shortest path in $G$. If there are $\ell$ edges of $\pi(s,t)$ absent from $H$, then there is a set $N$ of $\frac{d\ell}{6} = \Omega(d\ell)$ vertices, where each vertex in $N$ is adjacent to a vertex on $\pi(s,t)$, connected via an edge of weight at most $W(s,t)$.
\end{lemma}

We refer to vertices in $N$ as the $d$-\emph{light neighbors} of $\pi(s,t)$. The fact that $d$-light neighbors are connected to $\pi(s,t)$ via light edges of weight $\le W(s,t)$ was not explicitly stated in~\cite{ahmed2020weighted} but follows directly from the proof, as $N$ is constructed by taking the $d$ lightest edges incident to vertices on $\pi(s,t)$ which are incident to a missing edge of weight at most $W(s,t)$.
After $d$-light initialization, additional edges or paths are added to $H$ in order to ``satisfy'' the remaining unsatisfied vertex pairs\footnote{A vertex pair $(s,t)$ is \emph{satisfied} if the spanner inequality~\eqref{eqn:spanner-ineq} holds for that pair.}.
We will use the standard method of $+\beta$ \emph{spanner completion}, where we iterate over each vertex pair $(s,t)$ in nondecreasing order of maximum weight $W(s,t)$ (then by nondecreasing distance $d_G(s,t)$ in case of a tie) and add $\pi(s,t)$ to the spanner if $(s,t)$ is unsatisfied.

\section{Sparse local additive spanners} \label{sec:pairwise}
Instead of $+cW$ additive error considered in~\cite{ahmed2020weighted}, we consider $+cW(\cdot, \cdot)$ local error.

\subsection{Pairwise $+(2+\eps)W(\cdot,\cdot)$ and $+(6+\eps)W(\cdot,\cdot)$ spanners}
For the pairwise $+(2+\eps)W(\cdot,\cdot)$ and $+(6+\eps)W(\cdot,\cdot)$ spanners (Theorem~\ref{thm:sparse}.\ref{thm:pairwise-2epsW}-\ref{thm:pairwise-6epsW}), we describe a deterministic construction. 
The analysis behind the edge bounds uses a set-off and improving strategy also used in~\cite{elkin2020improved}. The constructions involve one additional step of adding a certain number of edges along every vertex pair's shortest path before spanner completion.

Let $\ell, d \ge 1$ be parameters which are defined later. Let $H$ be a $d$-light initialization. Then for each vertex pair $(s,t) \in \Pc$, consider the shortest path $\pi(s,t)$ in $G$ and add the first $\ell$ missing edges and the last $\ell$ missing edges to $H$ (if $\pi(s,t)$ is missing at most $2\ell$ edges, all missing edges from $\pi(s,t)$ are added to $H$). We remark that if $(s,t)$ is already satisfied, we can skip this step for the pair $(s,t)$.

After this phase, we perform $+(2+\eps)W(\cdot,\cdot)$ or $+(6+\eps)W(\cdot,\cdot)$ spanner completion as follows: for each pair $(s,t) \in \Pc$ sorted in nondecreasing order of maximum weight $W(s,t)$, if $(s,t)$ is still unsatisfied, add the remaining edges from the path $\pi(s,t)$ to $H$. This construction clearly outputs a valid pairwise spanner, and $O(nd + \ell |\Pc|)$ edges are added in the ``distance-oblivious'' phase before spanner completion. It remains to determine the number of edges added in spanner completion.

Let $(s,t) \in \Pc$ be a vertex pair for which $\pi(s,t)$ is added to $H$ during spanner completion. Observe that after $d$-light initialization, the first $\ell$ missing edges and the last $\ell$ missing edges from $\pi(s,t)$ do not overlap; otherwise no remaining edges from $\pi(s,t)$ would have been added to $H$. Let $u_1v_1$, \ldots, $u_{\ell}v_{\ell}$ denote the first $\ell$ missing edges on $\pi(s,t)$ which are added after $d$-light initialization, and let $u_1'v_1'$, \ldots, $u_{\ell}'v_{\ell}'$ denote the last $\ell$ missing edges, where $u_i$ (or $u_i'$) is closer to $s$ than $v_i$ (or $v_i'$); see Fig.~\ref{fig:6W-spanners} for illustration.


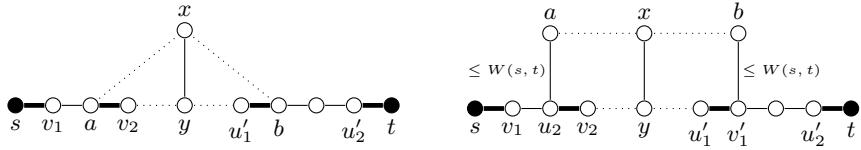
\begin{figure}
    \centering
    \input{figures/6W-spanners}
    \caption{Illustration of Lemma~\ref{lem:improvements-2epsW} (left) and Lemma~\ref{lem:improvements-6W} (right) with $\ell = 2$. Note that $s = u_1$ and $t = v_2'$ in this example. By adding $\pi(s,t)$, at least one of the pairs' $(a,x)$ or $(x,b)$ distance improves by at least $\frac{\eps W(s,t)}{2}$. Note that $a$, $x$, $b$ are not necessarily distinct.}
    \label{fig:6W-spanners}
\end{figure}


We will refer to the set $\{u_1, \ldots, u_{\ell}\}$ as the \emph{prefix} and the set $\{v_1', \ldots, v_{\ell}'\}$ as the \emph{suffix}. Consider the shortest paths $\pi(s,v_{\ell})$ and $\pi(u_1', t)$. By Lemma~\ref{lemma:d_initialize}, there are $\Omega(d \ell)$ $d$-light neighbors which are adjacent to a vertex in the prefix and suffix, respectively, connected by an edge of weight at most $W(s,t)$.

Consider the subpath $\pi(v_{\ell}, u_1')$; suppose $z \ge 1$ edges of $\pi(s,t)$ are added during spanner completion. By Lemma~\ref{lemma:d_initialize} again, there are $\Omega(dz)$ $d$-light neighbors which are adjacent to a vertex on $\pi(v_{\ell}, u_1')$. In the following lemmas, denote by $H_0$ and $H_1$ the spanner immediately before and after $\pi(s,t)$ is added, respectively.

\begin{lemma}\label{lem:improvements-2epsW}
Let $(s,t) \in \Pc$ be such that $\pi(s,t)$ is added to $H$ during $+(2+\eps)W(\cdot,\cdot)$ spanner completion. Let $a$ and $b$ be vertices in the prefix and suffix respectively. Let $x$ be a $d$-light neighbor of the path $\pi(v_{\ell}, u_1')$. Then both of the following hold:
\begin{enumerate}
\setlength{\itemsep}{5pt}
\item $d_{H_1}(a,x) \le d_G(a,x) + 2W(s,t)$ and $d_{H_1}(b,x) \le d_G(b,x) + 2W(s,t)$
\item $d_{H_0}(a,x) - d_{H_1}(a,x) > \frac{\eps W(s,t)}{2}$ or $d_{H_0}(b,x) - d_{H_1}(b,x) > \frac{\eps W(s,t)}{2}.$
\end{enumerate}
\end{lemma}

\begin{proof}
\begin{enumerate}
    \item We have by the triangle inequality
    \begin{align}
        d_{H_1}(a,x) &\le d_G(a,y) + w(xy) \le d_G(a,y) + W(s,t) \label{ineq:triangle-2W}\\
        &\le [d_G(a,x) + w(xy)] + W(s,t) \tag*{}\\
        &\le d_G(a,x) + 2W(s,t) \tag*{}
    \end{align}
    where we have used the fact that $x$ is a $d$-light neighbor, so $w(xy) \le W(s,t)$ by Lemma~\ref{lemma:d_initialize}. Similarly, $d_{H_1}(b,x) \le d_G(b,x) + 2W(s,t)$.
    
    \item The vertices $a$ and $b$ are on the prefix and suffix respectively. Suppose otherwise $d_{H_0}(a,x) - d_{H_1}(a,x) \le \frac{\eps W(s,t)}{2}$ and $d_{H_0}(x,b) - d_{H_1}(x,b) \le \frac{\eps W(s,t)}{2}$. Consider a shortest $s$--$t$ path in $H_0$ if it exists (if no such path exists, then $d_{H_0}(s,x) = \infty$ or $d_{H_0}(t,x) = \infty$, and 2. follows). Since $a$ is on the prefix and the first $\ell$ edges of $\pi(s,t)$ are added before completion, we necessarily have $d_{H_0}(s,a) = d_G(s,a)$ and symmetrically $d_{H_0}(b,t) = d_G(b,t)$. Then
    
\begin{align*}
    d_{H_0}(s,t) &\le d_G(s,a) + d_{H_0}(a,x) + d_{H_0}(x,b) + d_G(b,t) \\
    &\le d_G(s,a) + \left[d_{H_1}(a,x) + \frac{\eps W(s,t)}{2}\right] \\
    &\quad + \left[d_{H_1}(x,b) + \frac{\eps W(s,t)}{2}\right] + d_G(b,t)\\
    &\underset{\eqref{ineq:triangle-2W}}{\le} d_G(s,a) + d_G(a,y) + d_G(y,b) + d_G(b,t) + (2+\eps)W(s,t) \\
    &= d_G(s,t) + (2+\eps)W(s,t)
\end{align*}
    contradicting that $\pi(s,t)$ was added during $+(2+\eps)W(\cdot,\cdot)$ spanner completion.
\end{enumerate}\qed
\end{proof}

For the pairwise $+(6+\eps)W(\cdot,\cdot)$ spanner, we consider arbitrary $d$-light neighbors $a$ and $b$ of the prefix and suffix, and similarly consider vertex pairs $(a,x)$, $(b,x)$ whose distances sufficiently improve:
\begin{lemma} \label{lem:improvements-6W}
Let $(s,t) \in \Pc$ be such that $\pi(s,t)$ is added to $H$ during $+(6+\eps)W(\cdot,\cdot)$ spanner completion. Let $a$ and $b$ be $d$-light neighbors adjacent to vertices $u_i$ and $v_j'$ in the prefix and suffix, respectively. Let $x$ be a $d$-light neighbor of the path $\pi(v_{\ell}, u_1')$. Then both of the following hold:

\begin{enumerate}
\setlength{\itemsep}{5pt}
\item $d_{H_1}(a,x) \le d_G(a,x) + 4W(s,t)$ and $d_{H_1}(x,b) \le d_G(x,b) + 4W(s,t)$
\item $d_{H_0}(a,x) - d_{H_1}(a,x) > \frac{\eps W(s,t)}{2}$ or $d_{H_0}(x,b) - d_{H_1}(x,b) > \frac{\eps W(s,t)}{2}.$
\end{enumerate}
\end{lemma}


\begin{proof}
\begin{enumerate}
\item Similar to Lemma~\ref{lem:improvements-2epsW}, we have by the triangle inequality
\begin{align}
d_{H_1}(a,x) &\le w(au_i) + d_G(u_i, y) + w(xy) \tag*{} \\
&\le d_G(u_i, y) + 2W(s,t) \label{ineq:triangle-6W}\\
&\le [w(u_ia) + d_G(a,x) + w(xy)] + 2W(s,t) \tag*{} \\
&\le d_G(a,x) + 4W(s,t) \tag*{}
\end{align}
where we have used the fact that $a$ and $x$ are $d$-light neighbors, so $w(u_ia) \le W(s,t)$ and $w(xy) \le W(s,t)$. Similarly, $d_{H_1}(x,b) \le d_G(x,b) + 4W(s,t)$.
\item Suppose otherwise $d_{H_0}(a,x) - d_{H_1}(a,x) \le \frac{\eps W(s,t)}{2}$ and $d_{H_0}(x,b) - d_{H_1}(x,b) \le \frac{\eps W(s,t)}{2}$. Consider a shortest $s$--$t$ path in $H_0$ if it exists (if no such path exists, then $d_{H_0}(s,x) = \infty$ or $d_{H_0}(t,x) = \infty$ and 2. follows). Then
\begin{align*}
d_{H_0}(s,t) &\le d_G(s,u_i) + w(u_ia) + d_{H_0}(a,x) + d_{H_0}(x,b) + w(bv_i') + d_G(v_i't) \\
&\le d_G(s,u_i) + W(s,t) + \left[d_{H_1}(a,x) + \frac{\eps W(s,t)}{2}\right] \\
&\quad +\left[d_{H_1}(x,b) + \frac{\eps W(s,t)}{2}\right] + W(s,t) + d_G(v_i't) \\
&\le d_G(s,u_i) + d_{H_1}(a,x) + d_{H_1}(x,b) + d_G(v_i',t) + (2+\eps)W(s,t) \\
&\underset{\eqref{ineq:triangle-6W}}{\le} d_G(s,u_i) + [d_G(u_i,y) + 2W(s,t)] \\
&\quad + [d_G(y,v_i') + 2W(s,t)] + d_G(v_i',t) + (2+\eps)W(s,t) \\
&= d_G(s,t) + (6+\eps)W(s,t)
\end{align*}
contradicting that $\pi(s,t)$ was added during $+(6+\eps)W(\cdot,\cdot)$ spanner completion.
\end{enumerate}\qed
\end{proof}

\begin{lemma} \label{lem:pairwise-size-2epsW}
By setting $d = |\Pc|^{1/3}$ and $\ell = n/|\Pc|^{2/3}$, the pairwise $+(2+\eps)W(\cdot,\cdot)$ construction outputs a subgraph $H$ with $|E(H)| = O_{\eps}\left(n|\Pc|^{1/3}\right)$.
\end{lemma}
\begin{proof}
In the distance-oblivious phase, we add $O(nd + \ell|\Pc|) = O(n|\Pc|^{1/3})$ edges to $H$.
A vertex pair $(v,x)$ is \emph{set-off} if it is the first time that $d_H(v,x) \le d_G(v,x) + (2+\eps)W(s,t)$ and is \emph{improved} if its distance in $H$ decreases by at least $\frac{\eps W(s,t)}{2}$. Suppose adding $\pi(s,t)$ during $+(2+\eps)W(\cdot,\cdot)$ spanner completion adds $z \ge 1$ additional edges. Let $x$ be a $d$-light neighbor of $\pi(v_{\ell},u_1')$ and let $a$, $b$ be vertices in the prefix and suffix. By Lemma~\ref{lemma:d_initialize}, there are $\Omega(dz)$ vertices $x$ adjacent to $\pi(v_{\ell}u_1')$. By Lemma~\ref{lem:improvements-2epsW}, both of the pairs $(a,x)$, $(b,x)$ are set-off if not already, and at least one of the pairs is improved upon adding $\pi(s,t)$. Since there are $\Omega(\ell)$ choices for $a$ or $b$, this gives $\Omega(dz \times \ell) = \Omega(nz/|\Pc|^{1/3})$ improvements upon adding $z$ edges to $H$. 

Once a pair $(v,x)$ is set-off, it can only be improved $O(\frac{1}{\eps})$ times; this follows since pairs are ordered by their maximum weight, so any improvement is by at least $\frac{\eps W(s,t)}{2}$. If $Z$ total edges are added during spanner completion, then the number of improvements is $\Omega(d\ell Z)$. There are $O(n^2)$ vertex pairs and once set-off, each vertex pair is improved $O\left(\frac{1}{\eps}\right)$ times, in which we have $ \Omega(d\ell Z) = O\left(\frac{n^2}{\eps}\right). $
Since $d = |\Pc|^{1/3}$ and $\ell = n/|\Pc|^{2/3}$, we obtain $Z = O\left(\frac{1}{\eps} n|\Pc|^{1/3}\right)$. Altogether we obtain $|E(H)| = O\left(\frac{1}{\eps}n|\Pc|^{1/3}\right)$.
\qed
\end{proof}

\begin{lemma}\label{lem:pairwise-size-6epsW}
By setting $d = |\Pc|^{1/4}$ and $\ell = n/|\Pc|^{3/4}$, the pairwise $+(6+\eps)W(\cdot,\cdot)$ construction outputs a subgraph $H$ with $|E(H)| = O_{\eps}(n|\Pc|^{1/4})$.
\end{lemma}
\begin{proof}
This is the same as Lemma~\ref{lem:pairwise-size-2epsW} except that $a$ and $b$ are $d$-light neighbors of the subpaths containing the first and last $\ell$ missing edges. Then if $z \ge 1$ edges in $\pi(s,t)$ are added during $+(6+\eps)W(\cdot,\cdot)$ spanner completion, we obtain $\Omega(dz \times d\ell) = \Omega(d^2 \ell z) = \Omega(nz/|\Pc|^{1/4})$ improvements upon adding $z$ edges to $H$.

If $Z$ total edges are added during $+(6+\eps)W(\cdot,\cdot)$ spanner completion, then similar reasoning gives
\[ \Omega(d^2 \ell Z) = O\left(\frac{n^2}{\eps}\right) \]
in which $Z = O(\frac{1}{\eps}n|\Pc|^{1/4})$ and $|E(H)| = O(\frac{1}{\eps}n|\Pc|^{1/4})$. \qed
\end{proof}

\begin{proof}[Lemma~\ref{lemma:2w}]
Since $x$ is a $d$-light neighbor of $\pi(s,t)$, we have $w(xy) \le W(s,t)$. Using the triangle inequality:
\begin{align*}
    d_G(s,x) + d_G(x,t) &\le [d_G(s,y) + w(xy)] + [w(xy) + d_G(y,t)] \\
    &\le d_G(s,t) + 2w(xy) \\
    &\le d_G(s,t) + 2W(s,t)
\end{align*}
completing the proof. \qed
\end{proof}

Lemmas~\ref{lem:pairwise-size-2epsW} and~\ref{lem:pairwise-size-6epsW} imply Theorem~\ref{thm:sparse}.\ref{thm:pairwise-2epsW} and~\ref{thm:sparse}.\ref{thm:pairwise-6epsW} respectively. Further, by setting $W=1$ and $\eps = 0.5$, these results imply pairwise $+2$ and $+6$ spanners of size $O(n|\Pc|^{1/3})$ and $O(n|\Pc|^{1/4})$ in \emph{unweighted} graphs (as a $+2.5$ spanner of an unweighted graph is also a $+2$ spanner). These edge bounds match those of existing pairwise $+2$ spanners~\cite{Kavitha15,censor2016distributed} and pairwise $+6$ spanners~\cite{Kavitha2017}.

\subsection{Pairwise $+2W(\cdot,\cdot)$ and $+4W(\cdot,\cdot)$ spanners}
We prove Theorem~\ref{thm:sparse}.\ref{thm:pairwise-2W} and~\ref{thm:sparse}.\ref{thm:pairwise-4W}: every weighted graph has a pairwise $+2W(\cdot,\cdot)$ spanner  and $+4W(\cdot,\cdot)$ spanner on $O(n|\Pc|^{1/3})$ edges and $O(n|\Pc|^{2/7})$ edges respectively. This removes the $\eps$ dependence from Theorem~\ref{thm:sparse}.\ref{thm:pairwise-2epsW} and uses local error instead of global $W$ error as in~\cite{ahmed2020weighted}. First, we need the following simple lemma:
\begin{lemma}\label{lemma:2w}
Let $H$ be a $d$-light initialization, and let $s, t \in V$. Let $x$ be a $d$-light neighbor of $\pi(s,t)$ connected to a vertex $y \in \pi(s,t)$ in $H$. Consider a shortest path tree in $G$ rooted at $x$. Then the distance from $s$ to $t$ in this tree is at most $d_G(s,t) + 2W(s,t)$.
\end{lemma}


\begin{proof}
Since $x$ is a $d$-light neighbor of $\pi(s,t)$, we have $w(xy) \le W(s,t)$. Using the triangle inequality:
\begin{align*}
    d_G(s,x) + d_G(x,t) &\le [d_G(s,y) + w(xy)] + [w(xy) + d_G(y,t)] \\
    &\le d_G(s,t) + 2w(xy) \\
    &\le d_G(s,t) + 2W(s,t)
\end{align*}
completing the proof. \qed
\end{proof}


\begin{proof}[Theorem~\ref{thm:sparse}.\ref{thm:pairwise-2W}]
The main idea is to first show the existence of a pairwise $+2W(\cdot,\cdot)$ \emph{spanner with slack} on $O(n|\Pc|^{1/3})$ which satisfies at least a constant fraction of the vertex pairs. Then, we apply a lemma from~\cite{bodwin2020note} (Lemma~\ref{lem:slack}) to show the existence of a pairwise spanner. Let $\ell, d \ge 1$ be parameters, and let $H$ be a $d$-light initialization ($O(nd)$ edges). We iterate over each pair $(s,t) \in \Pc$. If $d_H(s,t) \le d_G(s,t) + 2W(s,t)$, then $(s,t)$ is already satisfied and we can do nothing. If $d_H(s,t) > d_G(s,t) + 2W(s,t)$ and $\pi(s,t)$ has fewer than $\ell$ missing edges in $H$, add all missing edges in $\pi(s,t)$ to $H$ ($O(\ell |\Pc|)$ edges). Otherwise $(s,t)$ is unsatisfied and $\pi(s,t)$ is missing more than $\ell$ edges. In this case, we do not add any edges yet.

After iterating over each vertex pair, consider the following randomized process: for each vertex in $V$, select it with probability $\frac{6}{d\ell}$. If more than $\frac{6n}{d\ell}$ vertices are selected, repeat this process. For each selected vertex, add a shortest path tree rooted at that vertex to $H$. Note that at most $\frac{6n}{d\ell}$ vertices are selected, so at most $\frac{6n(n-1)}{d\ell} = O\left(\frac{n^2}{d\ell}\right)$ edges are added in this process.

By Lemma~\ref{lemma:d_initialize}, every unsatisfied pair $(s,t)$ has at least $\frac{d\ell}{6}$ $d$-light neighbors. Then with probability at least $1 - \left(1 - \frac{6}{d\ell}\right)^{d\ell / 6}$, we will have selected at least one $d$-light neighbor of $\pi(s,t)$. This probability is at least a constant, since $1 - \frac{1}{e} \le 1 - \left(1 - \frac{1}{x}\right)^x \le 1$ for all $x \ge 1$. By Lemma~\ref{lemma:2w}, adding a shortest path tree at any $d$-light neighbor of $\pi(s,t)$ satisfies the pair $(s,t)$.

Under this randomized process, every unsatisfied vertex pair is satisfied with at least constant probability $\alpha$, so by linearity of expectation, the expected number of satisfied vertex pairs is at least $\alpha |\Pc|$. Then there necessarily exists a subset of $\frac{6n}{d\ell}$ vertices such that adding a shortest path tree rooted at each vertex satisfies $\alpha |\Pc|$ vertex pairs.

The number of edges in this spanner with slack is at most $k\left(nd + \ell |\Pc| + \frac{n^2}{\ell d}\right)$ for some constant $k$; set $\ell := \frac{n}{|\Pc|^{2/3}}$ and $d := |\Pc|^{1/3}$ to obtain the desired edge bound of $kn|\Pc|^{1/3} = O(n|\Pc|^{1/3})$. This only implies the existence of a $+2W(\cdot,\cdot)$ spanner which satisfies a constant fraction $\alpha$ of vertex pairs; to finish, we use the following lemma: 
\begin{lemma}[\cite{bodwin2020note}] \label{lem:slack}
Let $a, b, c, k, \alpha > 0$ be constants (independent of $n$ or $|\Pc|$) with $\alpha \le 1$, let $G$ be an $n$-vertex graph, and let $p^*$ be a parameter. Suppose that:
\begin{itemize}
    \item There is an algorithm which, on input $G$, $\Pc$ with $|\Pc| \le p^*$, returns a pairwise spanner on $O(n^a)$ edges
    \item There is an algorithm which, on input $G$, $\Pc$ with $|\Pc| \ge p^*$, returns a spanner with slack on $\le k n^b|\Pc|^c$ edges which satisfies at least $\alpha |\Pc|$ vertex pairs
\end{itemize}
Then there is a pairwise spanner of $G$, $\Pc$ on $O(n^a + n^b |\Pc|^c)$ edges.
\end{lemma}
\begin{proof}[Lemma~\ref{lem:slack}]
While the number of unsatisfied vertex pairs $|\Pc|$ is greater than $p^*$, compute a spanner with slack on $\le k n^b |\Pc|^c$ edges, then remove the satisfied vertex pairs from $\Pc$. Repeat this process until $|\Pc| \le p^*$, then compute a pairwise spanner on $O(n^a)$ edges. The union of these computed spanners is a pairwise spanner of $G$ over the initial set of vertex pairs $\Pc$.

Let $p_1 = |\Pc|$ denote the number of vertex pairs initially, and for $i \ge 1$, let $p_i$ denote the number of unsatisfied vertex pairs immediately before the $i^{\text{th}}$ iteration. Since each spanner with slack satisfies at least a constant fraction $\alpha$ of the unsatisfied vertex pairs, we have $p_i \le (1-\alpha)p_{i-1}$, which implies $p_i \le p_1(1-\alpha)^{i-1}$ for all $i \ge 1$. Let $H_i = (V_i, E_i)$ denote the spanner with slack computed on the $i^{\text{th}}$ iteration, so that $|E_i| \le kn^b p_i^c$. By a simple union bound, we have

\begin{align*} 
\left|\bigcup_{i \ge 1} E_i\right| &\leq  \sum_{i \ge 1} |E_i| \le  \sum_{i \ge 1} kn^b p_i^c \\
&\le kn^b \sum_{i \ge 1}  (p_1(1-\alpha)^{i-1})^c \\
&= kn^b {p_1}^c \sum_{i \ge 1} ((1-\alpha)^c)^{i-1} \\
&= O(n^b |\Pc|^c)
\end{align*}
since $\sum_{i \ge 1} ((1-\alpha)^c)^{i-1}$ is a geometric series with common ratio $(1-\alpha)^c < 1$. The final pairwise spanner on $\le p^*$ pairs adds $O(n^a)$ edges, completing the proof.
\qed
\end{proof}

Lemma~\ref{lem:slack} implies the existence of a pairwise $+2W(\cdot,\cdot)$ spanner on $O(n|\Pc|^{1/3})$ edges.

\qed
\end{proof}

We now show that every weighted graph $G$ has a pairwise $+4W(\cdot,\cdot)$ spanner on $O(n|\Pc|^{2/7})$ edges.
\begin{lemma}\label{lemma:4w}
Let $H$ be a $d$-light initialization and let $s, t \in V$. Let $v_1$ and $v_2$ be two different $d$-light neighbors of $\pi(s,t)$, connected to two (not necessarily distinct) vertices $u_1, u_2 \in \pi(s,t)$ respectively. Consider adding the missing edges along the shortest paths from $s$--$u_1$, $v_1$--$v_2$, and $u_2$--$t$ to $H$. Then $d_H(s,t) \le d_G(s,t) + 4W(s,t)$.
\end{lemma}
\begin{proof}
See Fig.~\ref{fig:4W-spanners} for illustration. Since $v_1$ and $v_2$ are $d$-light neighbors of $u_1$, $u_2$ respectively, we have $w(u_1v_1) \le W(s,t)$ and $w(u_2v_2) \le W(s,t)$. By triangle inequality, we have
\begin{align*}
    d_H(s,t) &\le d_H(s,u_1) + w(u_1v_1) + d_H(v_1,v_2) + w(u_2v_2) + d_H(u_2,t) \\
    &\le d_G(s,u_1) + W(s,t) + d_G(v_1,v_2) + W(s,t) + d_G(u_2,t) \\
    &= d_G(s,u_1) + d_G(v_1, v_2) + d_G(u_2, t) + 2W(s,t)
    \intertext{Again by triangle inequality, $d_G(v_1,v_2) \le w(u_1v_1) + d_G(u_1,u_2)+w(u_2v_2) \le d_G(u_1,u_2) + 2W(s,t)$; hence we have}
    &\le d_G(s,u_1) + d_G(u_1,u_2) + d_G(u_2,t) + 4W(s,t) \\
    &= d_G(s,t) + 4W(s,t)
\end{align*}
as desired.
\qed
\end{proof}

\begin{proof}[Theorem~\ref{thm:sparse}.\ref{thm:pairwise-4W}]
As in Theorem~\ref{thm:sparse}.\ref{thm:pairwise-2W}, we use two parameters $d, \ell \ge 1$. Let $H$ be a $d$-light initialization. We iterate over each pair $(s,t) \in \Pc$. If $(s,t)$ is already satisfied, do nothing; otherwise, if $\pi(s,t)$ is missing at most $\ell$ edges in $H$, add the shortest path $\pi(s,t)$ to $H$. This adds at most $nd + \ell p$ edges to $H$.

Construct a random sample $R_1$ of vertices by selecting each vertex with probability $\frac{6d}{n}$; repeat if more than $6d$ vertices were selected. Add a shortest path tree rooted at each $v \in R_1$. We claim that if $(s,t)$ is unsatisfied and there are at least $\frac{n}{d^2}$ missing edges on $\pi(s,t)$, then with at least a constant probability, we will have selected a $d$-light neighbor of $\pi(s,t)$. Since $\pi(s,t)$ has at least $\frac{n}{d^2}$ missing edges, Lemma~\ref{lemma:d_initialize} implies there are at least $\frac{n}{6d}$ $d$-light neighbors of $\pi(s,t)$. Using a similar reasoning to Theorem~\ref{thm:sparse}.\ref{thm:pairwise-2W}, with at least a constant probability, we will have selected at least one $d$-light neighbor for $R_1$. By Lemma~\ref{lemma:2w}, adding the shortest path tree rooted at this neighbor satisfies the pair $(s,t)$. This sample $R_1$ adds $O(dn)$ edges to $H$.

Lastly, there may be vertex pairs $(s,t)$ such that $(s,t)$ is unsatisfied, and $\pi(s,t)$ is missing greater $\ell$ but fewer than $\frac{n}{d^2}$ edges. For each such pair, add the first $\ell$ missing edges and the last $\ell$ missing edges to $H$ ($O(\ell p))$ edges). Let the \emph{prefix} denote the set of vertices on $\pi(s,t)$ which are incident to one of the first $\ell$ missing edges which were added to $H$; define the \emph{suffix} similarly.

\begin{figure}[h]
\centering
\input{figures/4W-spanners}
\caption{Illustration of the $+4W(\cdot,\cdot)$ randomized construction, with edges incident to the prefix and suffix bolded ($\ell=2$ in this example).}
\label{fig:4W-spanners}
\end{figure}
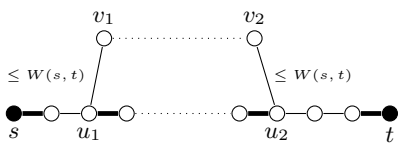

Construct a second random sample $R_2$ by sampling each vertex with probability $\frac{6}{d\ell}$; repeat if more than $\frac{6n}{d\ell}$ vertices are sampled. Then, for every two vertices $v_1, v_2 \in R_2$, add all edges among the shortest $v_1$-$v_2$ path among all paths which are missing at most $\frac{n}{d^2}$ edges; this adds $O\left(\left(\frac{n}{d\ell}\right)^2 \cdot \frac{n}{d^2}\right) = O\left(\frac{n^3}{\ell^2 d^4}\right)$ edges. By Lemma~\ref{lemma:d_initialize}, there are $\frac{d\ell}{6} = \Omega(d\ell)$ $d$-light neighbors of the prefix, and $\Omega(d\ell)$ $d$-light neighbors of the suffix. Then with at least a constant probability, we will have sampled in $R_2$ at least one $d$-light neighbor $v_1$ of the prefix and one $d$-light neighbor $v_2$ of the suffix of $\pi(s,t)$. By adding a shortest path from $v_1$ to $v_2$ containing at most $\frac{n}{d^2}$ edges along with the first and last $\ell$ edges of $\pi(s,t)$, this satisfies the pair $(s,t)$ by Lemma~\ref{lemma:4w}. 


The number of edges in $H$ is $|E(H)| = O\left(nd + |\Pc|\ell + n^3/(\ell^2d^4) \right)$. Set $\ell := n/|\Pc|^{5/7}$ and $d := |\Pc|^{2/7}$ to obtain the desired edge bound of $O(n|\Pc|^{2/7})$. As in Theorem~\ref{thm:sparse}.\ref{thm:pairwise-2W}, this only ensures the existence of a $+4W(\cdot,\cdot)$ spanner with slack; apply Lemma~\ref{lem:slack} to prove Theorem~\ref{thm:sparse}.\ref{thm:pairwise-4W}.
\qed
\end{proof}

\subsection{All-pairs $+cW(\cdot, \cdot)$ spanners}\label{sec:allpairs} 

We now discuss the all-pairs setting, to prove Theorem~\ref{thm:sparse}.\ref{thm:allpairs-4W}.
Our construction is essentially the same as the pairwise construction given above, but the analysis is slightly different.
The parameter $\ell$ no longer plays a role, and we skip the step where we consider pairs $(s, t)$ such that $\pi(s, t)$ is missing at most $\ell$ edges.
But the rest of the construction is the same:
\begin{itemize}
    \item We start the spanner $H$ as a $d$-light initialization,
    \item We take a random sample $R_1$ by selecting each vertex with probability $\frac{6d}{n}$ and add a shortest path tree rooted at each sampled vertex, and then
    \item We take a random sample $R_2$ by including each vertex with probability $\frac{6}{d}$.
    For any two nodes $s, t \in R_2$, let $\pi(s, t)$ be the shortest $s \leadsto t$ path among all paths that are missing $\le \frac{n}{d^2}$ edges in the initialization (if no such path exists, then ignore the pair $s, t$).
    Then we add all edges in $\pi(s, t)$ to the spanner.
\end{itemize}
The total number of edges in this construction is $O(nd + \frac{n^3}{d^4})$.  Setting $d = n^{2/5}$ gives an edge bound of $O(n^{7/5})$.
Finally, by the same analysis as above, each demand pair is satisfied with constant probability.
So if we repeat the construction $C \log n$ times for a large enough constant $C$, and we union together the $O(\log n)$ spanners computed in each round, then with high probability all demand pairs are satisfied.
The final spanner thus has $O(n^{7/5} \cdot \log n) = \Oish(n^{7/5})$ edges in total.

\section{Lightweight local additive spanners} \label{sec:lightweight}
In this section, we prove Theorem~\ref{thm:lightness} by constructing \emph{lightweight} additive spanners. 
The previous spanner algorithms based on $d$-light initialization can produce spanners with poor lightness: Let $G$ be the graph obtained by taking $K_{\frac{n}{2},\frac{n}{2}}$ with all edges of weight $W$, then adding two paths of weight 0 connecting the vertices within each bipartition. Then $w(\MST(G)) = W$, while $d$-light initialization already adds $\Omega(Wnd)$ weight to the spanner $H$, or $\Omega(nd)$ lightness.

In order to construct lightweight additive spanners, we introduce a new initialization technique called $d$-\emph{lightweight initialization}, which adds edges of total weight at most $d$ for each vertex, starting from the MST of $G$. We first perform the following simple modifications to the input graph $G$ in order:
\begin{enumerate}\setlength{\itemsep}{5pt}
\item Scale the edge weights of $G$ linearly so that the weight of the MST is $\frac{n}{2}$.

\item Remove all edges of weight $\ge n$ from $G$.
\end{enumerate}
Note that step 1 also scales $W(s,t)$ for vertex pairs $(s,t)$, but the validity of an additive spanner or spanner path is invariant to scaling. Step 2 does not change the shortest path metric or the maximum edge weights $W(s,t)$; since the MST has weight $\frac{n}{2}$, this implies $W(s,t) \le d_G(s,t) \le \frac{n}{2}$.

Given $G$ and $d \ge 0$, the $d$-\emph{lightweight initialization} is a subgraph $H$ defined as follows: let $H$ be an MST of $G$. If $d > 0$, then for each vertex $v \in V$, consider all incident edges to $v$ which have not already been added to $H$, in nondecreasing order of weight. Add these incident edges one by one to $H$ until the next edge causes the total edge weight added corresponding to $v$ to be greater than $d$. This subgraph $H$ has $O(n + nd)$ total edge weight, and $O(1+d)$ lightness. Note that 0-lightweight initialization is simply the MST of $G$. After $d$-lightweight initialization, perform $+cW(\cdot,\cdot)$ spanner completion by iterating over each vertex pair $(s,t)$ in nondecreasing order of the maximum weight $W(s,t)$ (then by nondecreasing distance $d_G(s,t)$ if two pairs have the same maximum weight), and adding $\pi(s,t)$ to $H$ if $(s,t)$ is unsatisfied.

\subsection{Lightweight initialization and neighborhoods of shortest paths}
In order to prove the desired lightness bounds, we consider a subdivided graph obtained as follows: for each MST edge $e$ of $\MST(G)$, subdivide $e$ into $\lceil w(e) \rceil$ edges of weight $\frac{w(e)}{\lceil w(e) \rceil}$. Recall that $w(\MST(G)) = \frac{n}{2}$, so this adds $O(n)$ vertices, all of which are on MST edges. Let $G'$ denote this subdivided graph. We do not modify the maximum edge weights $W(s,t)$ for each pair $(s,t) \in V(G)$, even if such edges are subdivided in $G'$. We remark that the MST of $G'$ also has weight $\frac{n}{2}$, and all MST edges have weight at most 1.

\begin{lemma}\label{lemma:lightweight-neighborhood-1}
Let $e = uv$ be an edge in $E(G')$ which is not contained in the $d$-lightweight initialization $H$. Then there are $\Omega(d^{1/2})$ vertices $x$ in $G'$ such that $d_H(u,x) \le w(e)$. Moreover, if $\eps' \in (0,1]$, there are $\Omega(\eps' w(e))$ vertices $x$ such that $d_H(u,x) \le \eps' w(e)$.
\end{lemma}
\begin{proof}
If $w(e) \ge d^{1/2}$, then consider the set $N$ of vertices (including $u$) such that each $x \in N$ is connected to $u$ by a path of $\le \lfloor w(e) \rfloor$ MST edges.
Since all MST edges have weight $\le 1$, all vertices $x \in N$ satisfy $d_H(u, x) \le w(e)$. Then $|N| = \max\{1, \lfloor w(e) \rfloor\} \ge \frac{w(e)}{2} = \Omega(d^{1/2})$, proving this case. Note that $w(e) \le n$ by step 2 above. Moreover given $\eps' \in (0,1]$, consider vertices $x$ which are at most $\lfloor \eps' w(e)\rfloor$ MST edges away from $u$. There are $\max\{1, \lfloor \eps' w(e) \rfloor\} \ge \frac{\eps' w(e)}{2} = \Omega(\eps' w(e))$ such vertices, all of distance $\le \eps' w(e)$ from $u$.

If $w(e) < d^{1/2}$, then every edge added when $u$ is considered in the $d$-lightweight initialization has weight less than $d^{1/2}$, since edges are added in nondecreasing weight.
We thus add at least $d^{1/2} - 1$ edges corresponding to $u$. By construction, each such edge is lighter than $e$, so for any vertex $x$ that is the endpoint of one of these edges $ux$, we have $d_H(u, x) \le w(ux) \le w(e)$. Lastly by considering $x=u$, we obtain $d^{1/2}$ vertices $x$ in $G'$. \qed
\end{proof}

We now prove a lightweight analogue of Lemma~\ref{lemma:d_initialize}:
\begin{lemma}\label{lemma:lightweight-neighborhood}
Let $H$ be a $d$-lightweight initialization, let $\pi(s,t)$ be a shortest $s$--$t$ path in $G'$, let $z$ be the total weight of the edges in $\pi(s,t)$ which are absent in $H$, and let $\eps' \in (0,1]$. Then there is a set $N$ of at least $\frac{\eps'z}{10} = \Omega(\eps'z)$ vertices in $G'$ which are of distance at most $\eps'W(s,t)$ in $H$ from some vertex in $\pi(s,t)$.
\end{lemma}
\begin{proof}
Let $M_0 = \{e_1, \ldots, e_{\ell}\}$ be the set of $\ell$ edges in $\pi(s,t)$ which are absent in $H$, where $z$ is the total weight of these edges. We denote the endpoints of $e_i$ by $u_i$ and $v_i$ where $u_i$ is closer to $s$. We will first construct a subset $M_1 \subseteq M_0$ of edges which satisfies two properties: i) the neighborhoods of nearby MST vertices from each edge in $M_1$ are pairwise disjoint, and ii) the total edge weight in $M_1$ is $\Omega(z)$.

Let $M_1 := \emptyset$. While $M_0$ is not empty, select an edge $e_i \in M_0$ of maximum weight and add it to $M_1$. Then remove $e_i$ from $M_0$, as well as all other edges $e_j$ such that $d_{G'}(u_i,u_j) \le 2\eps' w(e_i)$ from $M_0$. When $M_0$ is empty, we will be left with a set $M_1$ of edges which are relatively far apart. For each edge $e_i = u_iv_i \in M_1$, add to $N$ all vertices connected to $u_i$ by a path of $\le \eps' w(e_i)$ MST edges (including $u_i$ itself).
Let $N_i$ denote the set of vertices added to $N$ corresponding to edge $e_i$.
Since MST edges in $G'$ have weight at most 1, we have $|N_i| \ge \frac{\eps'w(e_i)}{2}$ by Lemma~\ref{lemma:lightweight-neighborhood-1}, and $d_H(u_i,x) \le \eps' w(e_i) \le \eps'W(s,t)$ for each $x \in N_i$.


\begin{figure}[h]
\centering
\input{figures/initialization-lightness}
\caption{Illustration of Lemma~\ref{lemma:lightweight-neighborhood} where MST edges are bolded and $\eps'=\frac{1}{2}$. By selecting $e_i = u_iv_i$ of weight 5.2 and adding $e_i$ to $M_1$, the five MST vertices which are connected by at most $\lfloor \frac{1}{2} \times 5.2 \rfloor = 2$ MST edges are added to $N$, and $|N_i|=5$. Note that $u_{i+1} = v_i$ in this example since $e_i$ and $e_{i+1}$ are adjacent.}
\end{figure}
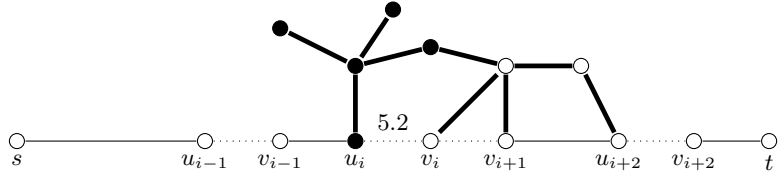


We claim that the sets $N_i : e_i \in M_1$ are pairwise disjoint. Suppose otherwise there exists some vertex $x \in N_i \cap N_j$, and w.l.o.g. edge $e_i$ was selected before edge $e_j$ so that $w(e_i) \ge w(e_j)$. Since $x \in N_i$, we have $d_H(u_i,x) \le \eps'w(e_i)$ and similarly $d_H(u_j,x) \le \eps'w(e_j)$. Then by the triangle inequality: 
\begin{align*}
d_{G'}(u_i, u_j) \le d_H(u_i, u_j) &\le d_H(u_i, x) + d_H(x, u_j) \\
&\le \eps'w(e_i) + \eps'w(e_j) \le 2\eps'w(e_i). \\
\end{align*}
However this is a contradiction: once edge $e_i$ is selected for $M_1$, edge $e_j$ is also removed from $M_0$ since the distance from $u_i$ to $u_j$ is at most $2\eps'w(e_i)$, so $e_j$ could not have been selected for $M_1$.

We also claim that the total edge weight in $M_1$ is $\Omega(z)$.
Whenever an edge $e_i$ is selected for $M_1$, we remove $e_i$ and any nearby edges $e_j$ with $d_{G'}(u_i,u_j) \le 2\eps'w(e_i)$ from $M_0$. Since all nearby edges by definition are on $\pi(s,t)$, the total weight of the edges removed from $M_0$ when $e_i$ is selected is at most $(1+4\eps')w(e_i)$. Then the total edge weight in $M_1$ is at least $\frac{z}{1+4\eps'} \ge \frac{z}{5} = \Omega(z)$, since $0 < \eps' \le 1$. 

Let $N := \bigcup_{e_i \in M_1} N_i$.
As shown previously, each $N_i$ contains at least $\frac{\eps'w(e_i)}{2}$ vertices.
By the above claims, the neighborhoods $N_i$ are pairwise disjoint and the sum of edge weights in $M_1$ is at least $\frac{z}{5}$, so we obtain

\[
    |N| = \sum_{e_i \in M_1} |N_i| \ge \sum_{e_i \in M_1} \frac{\eps' w(e_i)}{2} \ge \frac{\eps'}{2} \cdot \frac{z}{5} = \Omega(\eps'z). \qquad \qed
\]
\end{proof}

\subsection{Proof of Theorem~\ref{thm:lightness}}
Using Lemmas~\ref{lemma:lightweight-neighborhood-1} and~\ref{lemma:lightweight-neighborhood}, we can analyze the lightness of the above spanner constructions by considering the number of vertex pairs in the subdivided graph $G'$ whose distances sufficiently improve upon adding $\pi(s,t)$, where $s,t \in V(G)$.

\begin{lemma}\label{lem:improvements-epsW}
Let $(s,t) \in V(G)$ be a vertex pair such that $\pi(s,t)$ is added during $+\eps W(\cdot,\cdot)$ spanner completion. Let $x \in N$ be a vertex in $G'$ which is of distance $\le \frac{\eps}{4}W(s,t)$ from some vertex $u$ in $\pi(s,t)$ in $H$. Let $H_0$ and $H_1$ denote the spanner before and after $\pi(s,t)$ is added. Then both of the following hold:
\begin{enumerate}
\item $d_{H_1}(s,x) \le d_{G'}(s,x) + \frac{\eps}{2}W(s,t)$ and $d_{H_1}(t,x) \le d_{G'}(t,x) + \frac{\eps}{2}W(s,t)$
\item $d_{H_0}(s,x) - d_{H_1}(s,x) \ge \frac{\eps}{4} W(s,t)$ or $d_{H_0}(t,x) - d_{H_1}(t,x) \ge \frac{\eps}{4} W(s,t).$
\end{enumerate}
\end{lemma}
Similar to Lemma~\ref{lem:improvements-2epsW}, statement 1. holds by the triangle inequality, and 2. can be proven by contradiction: if neither inequality was true, then $d_{H_0}(s,t) \le d_G(s,t) + \eps W(s,t)$, contradicting that $\pi(s,t)$ was added during spanner completion.

\begin{figure}[h]
\centering
\input{figures/epsW-spanners}
\caption{Illustration of Lemma~\ref{lem:improvements-epsW}. By adding $\pi(s,t)$ to $H$ during $+\eps W(\cdot,\cdot)$ spanner completion, both pairs $(s,x)$ and $(t,x)$ are satisfied, and at least one of the pairs' distances improves by at least  $\frac{\eps}{4}W(u,v)$.}
\label{fig:epsW-spanners}
\end{figure}
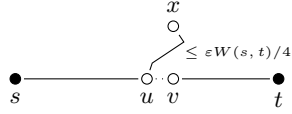
\begin{proof}[Theorem~\ref{thm:lightness}.\ref{thm:lightness-epsW}]
By Lemma~\ref{lemma:lightweight-neighborhood} with $\eps' = \frac{\eps}{4}$, there are $\Omega(\eps z)$ vertices $x \in V(G')$ which are of distance $\le \frac{\eps}{4}W(s,t)$ from some vertex in $\pi(s,t)$, so adding path $\pi(s,t)$ of weight $z$ improves $\Omega(\eps z)$ vertex pairs. If $Z$ is the total weight added during $+\eps W(\cdot,\cdot)$ spanner completion, then there are $\Omega(\eps Z)$ improvements.

Once a vertex pair is set-off, it is only improved a \emph{constant} number of times (since any such pair $(s,x)$ or $(t,x)$, once set-off, has error $\frac{\eps}{2}W(s,t)$, and any improvement is by $\Omega(\eps W(s,t))$). Then by considering the number of improvements, we obtain $\Omega(\eps Z) = O(n^2) \implies Z = O(\frac{1}{\eps}n^2)$. This result does not depend on $d$-lightweight initialization, so set $d = 0$. The total weight of the spanner $H$ is $O(\frac{1}{\eps}n^2)$; since $w(\MST(G))=\frac{n}{2}$, we obtain $\text{lightness}(H) = O_{\eps}(n)$ as desired. \qed
\end{proof}

For the $+(4+\eps)W(\cdot,\cdot)$ spanner, the corresponding lemma is as follows:

\begin{lemma}\label{lem:improvements-4epsW}
Let $(u,v) \in V(G)$ be a vertex pair such that $\pi(u,v)$ is added during $+(4+\eps)W(\cdot,\cdot)$ spanner completion. Let $a$ and $b$ be $d$-lightweight neighbors of $s$ and $t$ in $G'$, respectively, such that $d_{H_0}(x,a) \le W(s,t)$ and $d_{H_0}(y,b) \le W(s,t)$. Let $x \in N$ be a vertex in $G'$ which is of distance $\le \frac{\eps}{4}W(s,t)$ from some vertex $y$ in $\pi(s,t)$. Then both of the following hold:

\begin{enumerate}
\item $d_{H_1}(a,x) \le d_{G'}(a,x) + (2+\eps)W(s,t)$ and $d_{H_1}(b,x) \le d_{G'}(b,x) + (2+\eps)W(s,t)$
\item $d_{H_0}(a,x) - d_{H_1}(a,x) \ge \frac{\eps}{4}W(s,t)$ or $d_{H_0}(b,x) - d_{H_1}(b,x) \ge \frac{\eps}{4}W(s,t)$.
\end{enumerate}
\end{lemma}
Again, this is proved using the same methods as in Lemmas~\ref{lem:improvements-2epsW},~\ref{lem:improvements-6W}, and \ref{lem:improvements-epsW}; see Fig.~\ref{fig:4epsW-spanners}.

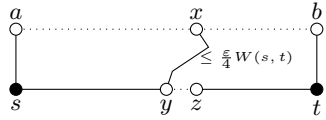
\begin{figure}[h]
\centering
\input{figures/4epsW-spanners}
\caption{Illustration of Lemma~\ref{lem:improvements-4epsW}. By adding $\pi(s,t)$ to $H$ during $+(4+\eps)W(\cdot,\cdot)$ spanner completion, both pairs $(a,x)$ and $(b,x)$ are satisfied, and at least one of the pairs' distances improves by at least $\frac{\eps}{4}W(s,t)$.}
\label{fig:4epsW-spanners}
\end{figure}

\begin{proof}[Theorem~\ref{thm:lightness}.\ref{thm:lightness-4epsW}]
Let $(s,t)$ be a vertex pair for which $\pi(s,t)$ is added during $+(4+\eps)W(\cdot,\cdot)$ spanner completion. Again, by Lemma~\ref{lemma:lightweight-neighborhood} with $\eps' = \frac{\eps}{4}$, there are $\Omega(\eps z)$ choices for $x$ which are of distance $\le \frac{\eps}{4}W(s,t)$ from some vertex in $\pi(u,v)$.

Similar to~\cite{elkin2020improved}, we observe that the first edge (say $st_1$) in $\pi(s,t)$ (starting from $s$) is absent in $H_0$ immediately before $\pi(s,t)$ is added. Suppose otherwise $st_1 \in E(H_0)$, then consider the pair $(t_1,t)$. Since $+(4+\eps)W(\cdot,\cdot)$ spanner completion considers all vertex pairs in nondecreasing $W(s,t)$ and then by distance $d_G(s,t)$, the pair $(t_1,t)$ is already satisfied before considering $(s,t)$. Then $d_{H_0}(s,t) \le w(st_1) + d_{H_0}(t_1,t) \le w(st_1) + [d_G(t_1,t) + (4+\eps)W(s,t)] \le  d_G(s,t) + (4+\eps)W(s,t)$, contradicting that $\pi(s,t)$ was added to $H$. Symmetrically, the last edge in $\pi(s,t)$ is absent in $H_0$.

By Lemma~\ref{lemma:lightweight-neighborhood-1} and the above observation, we can establish there are $\Omega(d^{1/2})$ choices for $a$ and $b$. Then for every choice of $x$, $a$, $b$, adding $\pi(s,t)$ sets off the pairs $(a,x)$, $(b,x)$ if not already, and improves at least one pair's distance by at least $\frac{\eps}{4}W(s,t)$. By Lemma~\ref{lemma:lightweight-neighborhood}, this gives $\Omega(\eps d^{1/2}z)$ improvements upon adding $z$ total edge weight in $\pi(s,t)$.

If $Z$ is the total weight added during $+(4+\eps) W(\cdot,\cdot)$ spanner completion, then
\[ \Omega(\eps d^{1/2}Z) = O\left(\frac{n^2}{\eps}\right) \implies Z = O_{\eps}\left(\frac{n^2}{d^{1/2}}\right).\]
By setting $d := n^{2/3}$, we obtain that the total weight of the spanner $H$ is $O_{\eps}(n^{5/3})$; since $w(\MST(G)) = \frac{n}{2}$, this implies $\text{lightness}(H) = O_{\eps}(n^{2/3})$ as desired.
 \qed
\end{proof}

\section{Conclusion}
We showed that many of the global $+cW$ spanners from~\cite{ahmed2020weighted} are actually local $+cW(\cdot,\cdot)$ spanners by strengthening the analysis and Lemma~\ref{lemma:d_initialize}. We additionally provided several additional pairwise and all-pairs spanners based on $d$-light initialization; this suggests $d$-light initialization can be used to construct several other types of additive spanners.

We then provided the first known lightweight additive spanner constructions, whose lightness does not depend on $W$. These lightness results are likely not optimal; we leave the questions of finding more optimized lightness guarantees, as well as lightness guarantees for pairwise and subsetwise spanners as future open problems. Further, another natural question is whether we can construct additive spanners in weighted graphs which are simultaneously sparse \emph{and} lightweight.

\paragraph{Acknowledgements} The authors wish to thank Michael Elkin, Faryad Darabi Sahneh, and the anonymous reviewers for their discussion and comments.

\bibliography{references}

\newpage
\appendix

\end{document}

%% file: figures/6W-spanners.tex
\begin{minipage}{0.49\textwidth}
\begin{tikzpicture}
    \node (s) at (0,0) {};
    \node (t) at (5,0) {};
    \node (v1) at (0.5,0) {};
    \node (v2) at (1,0) {};
    \node (v3) at (1.5,0) {};
    \node (v4) at (3,0) {};
    \node (v5) at (3.5,0) {};
    \node (v6) at (4,0) {};
    \node (v7) at (4.5,0) {};

    \node (x) at (2.25,1) {};
    \node (y) at (2.25,0) {};
    
    \foreach \x in {s,t}{
        \draw [fill=black] (\x) circle [radius=0.1];
    }
    \foreach \x in {s,t}{
        \node[below=2pt] at (\x) {$\x$};
    }
    
    \draw[dotted] (v3)--(v4);
    \foreach \x in {v1,v2,v3,v4,v5,v6,v7,x,y}
    	\draw[fill=white] (\x) circle [radius=0.1];
    
    \draw[line width=1.8pt] (s)--(v1);
    \draw[line width=1.8pt] (v2)--(v3);
    \draw[line width=1.8pt] (v4)--(v5);
    \draw[line width=1.8pt] (v7)--(t);
    \draw (v1)--(v2);
    \draw (v5)--(v6)--(v7);

    \draw (x)--(y);
    \draw[dotted] (v2)--(x)--(v5);

    \node[below=2pt] at (v1) {$v_1$};
    \node[below=2pt] at (v2) {$a$};
    \node[below=2pt] at (v3) {$v_2$};
    \node[below=2pt] at (v4) {$u_1'$};
    \node[below=2pt] at (v5) {$b$};
    \node[below=2pt] at (v7) {$u_2'$};
    \node[below=2pt] at (y) {$y$};
    \node[above=2pt] at (x) {$x$};
    
\end{tikzpicture}
\end{minipage}
\begin{minipage}{0.49\textwidth}
\begin{tikzpicture}
    \node (s) at (0,0) {};
    \node (t) at (5,0) {};
    \node (v1) at (0.5,0) {};
    \node (v2) at (1,0) {};
    \node (v3) at (1.5,0) {};
    \node (v4) at (3,0) {};
    \node (v5) at (3.5,0) {};
    \node (v6) at (4,0) {};
    \node (v7) at (4.5,0) {};
    
    \node (a) at (1,1) {};
    \node (b) at (3.5,1) {};
    \node (x) at (2.25,1) {};
    \node (y) at (2.25,0) {};
    
    \foreach \x in {s,t}{
        \draw [fill=black] (\x) circle [radius=0.1];
    }
    \foreach \x in {s,t}{
        \node[below=2pt] at (\x) {$\x$};
    }
    
    \draw[dotted] (v3)--(v4);
    \draw[dotted] (a)--(b);
    \foreach \x in {v1,v2,v3,v4,v5,v6,v7,a,b,x,y}
    	\draw[fill=white] (\x) circle [radius=0.1];
    
    \draw[line width=1.8pt] (s)--(v1);
    \draw[line width=1.8pt] (v2)--(v3);
    \draw[line width=1.8pt] (v4)--(v5);
    \draw[line width=1.8pt] (v7)--(t);
    \draw (v1)--(v2);
    \draw (v5)--(v6)--(v7);
    \draw (v2)--(a);
    \draw (b)--(v5);
    \draw (x)--(y);

    \node[below=2pt] at (v1) {$v_1$};
    \node[below=2pt] at (v2) {$u_2$};
    \node[below=2pt] at (v3) {$v_2$};
    \node[below=2pt] at (v4) {$u_1'$};
    \node[below=2pt] at (v5) {$v_1'$};
    \node[below=2pt] at (v7) {$u_2'$};
    \node[below=2pt] at (y) {$y$};
    \node[above=2pt] at (a) {$a$};
    \node[above=2pt] at (b) {$b$};
    \node[above=2pt] at (x) {$x$};
    \node[left=1pt] at (1.1,0.5) {\tiny $\le W(s,t)$};
    \node[right=1pt] at (3.4,0.5) {\tiny $\le W(s,t)$};
    
\end{tikzpicture}
\end{minipage}

%% file: figures/4W-spanners.tex
\begin{tikzpicture}
    \node (s) at (0,0) {};
    \node (t) at (5,0) {};
    \node (v1) at (0.5,0) {};
    \node (v2) at (1,0) {};
    \node (v3) at (1.5,0) {};
    \node (v4) at (3,0) {};
    \node (v5) at (3.5,0) {};
    \node (v6) at (4,0) {};
    \node (v7) at (4.5,0) {};
    
    \node (r1) at (1.2,1) {};
    \node (r2) at (3.2,1) {};
    
    \foreach \x in {s,t}{
        \draw [fill=black] (\x) circle [radius=0.1];
    }
    \foreach \x in {s,t}{
        \node[below=2pt] at (\x) {$\x$};
    }
    
    \foreach \x in {v1,v2,v3,v4,v5,v6,v7,r1,r2}
    	\draw[fill=white] (\x) circle [radius=0.1];
    
    \draw[line width=1.8pt] (s)--(v1);
    \draw[line width=1.8pt] (v2)--(v3);
    \draw[line width=1.8pt] (v4)--(v5);
    \draw[line width=1.8pt] (v7)--(t);
    \draw (v1)--(v2);
    \draw (v5)--(v6)--(v7);
    \draw (v2)--(r1);
    \draw (r2)--(v5);
    \draw[dotted] (v3)--(v4);
    \draw[dotted] (r1)--(r2);
    \node[below=2pt] at (v2) {$u_1$};
    \node[below=2pt] at (v5) {$u_2$};
    \node[above=2pt] at (r1) {$v_1$};
    \node[above=2pt] at (r2) {$v_2$};
    \node[left=1pt] at (1.1,0.5) {\tiny $\le W(s,t)$};
    \node[right=1pt] at (3.3,0.5) {\tiny $\le W(s,t)$};
    
\end{tikzpicture}

%% file: figures/initialization-lightness.tex
\begin{tikzpicture}
    \node (s) at (0,0) {};
    \node (t) at (10,0) {};
    \node (mst1) at (4.5,1) {};
    \node (mst2) at (3.5,1.5) {};
    \node (mst3) at (5,1.75) {};
    \node (mst4) at (7.5,1) {};
    \node (mst5) at (6.5,1) {};
    \node (mst6) at (5.5,1.25) {};
    \node (ui) at (4.5,0) {};
    \node (vi) at (5.5,0) {};
    \node (ua) at (2.5,0) {};
    \node (va) at (3.5,0) {};
    \node (vb) at (6.5,0) {};
    \node (uc) at (8,0) {};
    \node (vc) at (9,0) {};

    \node [below=2pt] at (ui) {$u_i$};
    \node [below=2pt] at (vi) {$v_i$};
    \node [above=1pt] at (5,0) {5.2};
    \node [below=2pt] at (ua) {$u_{i-1}$};
    \node [below=2pt] at (va) {$v_{i-1}$};
    \node [below=2pt] at (vb) {$v_{i+1}$};
    \node [below=2pt] at (s) {$s$};
    \node [below=2pt] at (t) {$t$};
    \node [below=2pt] at (uc) {$u_{i+2}$};
    \node [below=2pt] at (vc) {$v_{i+2}$};
    
    \draw[dotted] (ui)--(vi)--(vb);
    \draw[dotted] (ua)--(va);
    \draw[dotted] (uc)--(vc);
    \draw[line width=1.8pt] (ui)--(mst1)--(mst2);
    \draw[line width=1.8pt] (mst1)--(mst3);
    \draw[line width=1.8pt] (uc)--(mst4)--(mst5);
    \draw[line width=1.8pt] (mst5)--(mst6)--(mst1);
    \draw[line width=1.8pt] (vi)--(mst5)--(vb);
    \draw (s)--(ua);
    \draw (va)--(ui);
    \draw (vc)--(t);
    \draw (vb)--(uc);

    \foreach \x in {s,t,vi,ua,uc,va,vb,vc,mst4,mst5}{
    	\draw [fill=white] (\x) circle [radius=0.1];
    }
    \foreach \x in {ui,mst1,mst2,mst3,mst6}{
    	\draw [fill=black] (\x) circle [radius=0.1];
    }
\end{tikzpicture}

%% file: figures/epsW-spanners.tex
\begin{tikzpicture}[scale=0.70]
    \node (s) at (0,0) {};
    \node (t) at (5,0) {};
    
    \node (u) at (2.5,0) {};
    \node (x) at (3,1) {};
    \node (v) at (3,0) {};
    
    \foreach \x in {s,t}{
        \draw [fill=black] (\x) circle [radius=0.1];
    }
    \foreach \x in {u,v,x}{
    	\draw [fill=white] (\x) circle [radius=0.1];
    }
    \foreach \x in {s,t,u,v}{
        \node[below=2pt] at (\x) {$\x$};
    }
    \node[above=2pt] at (x) {$x$};
    
    \draw (s)--(u);
    \draw[dotted] (u)--(v);
    \draw (v)--(t);
    \draw (u)--(2.6,0.3)--(3.2,0.7)--(x);

    \node[right=1pt] at (3,0.5) {\tiny $\le \eps W(s,t)/4$};
\end{tikzpicture}

%% file: figures/4epsW-spanners.tex
\begin{tikzpicture}[scale=0.80,every node/.style={inner sep=0,outer sep=0}]
    \node (s) at (0,0) {};
    \node (t) at (5,0) {};
    \node (a) at (0,1) {};
    \node (b) at (5,1) {};
    
    \node (y) at (2.5,0) {};
    \node (x) at (3,1) {};
    \node (z) at (3,0) {};

    \foreach \x in {s,t,y,z}{
        \node[below=4pt] at (\x) {$\x$};
    }
    \foreach \x in {a,x,b}{
    	\node[above=4pt] at (\x) {$\x$};
    }
    
    \draw (s)--(y);
    \draw (s)--(a);
    \draw (t)--(b);
    \draw[dotted] (y)--(z);
    \draw[dotted] (a)--(x)--(b);
    \draw (z)--(t);
    \draw (y)--(2.6,0.3)--(3.2,0.7)--(x);
    
    \foreach \x in {s,t}{
        \draw [fill=black] (\x) circle [radius=0.1];
    }
    \foreach \x in {y,x,z,a,b}{
    	\draw [fill=white] (\x) circle [radius=0.1];
    }

    \node[right=1pt] at (3,0.5) {\tiny $\le \frac{\eps}{4}W(s,t)$};
\end{tikzpicture}